\documentclass{article}
\usepackage{ijcai26}
\usepackage{makecell}
\usepackage{empheq}
\usepackage{multirow}
\usepackage{times}
\usepackage{soul}
\usepackage{url}
\usepackage[hidelinks]{hyperref}
\usepackage[utf8]{inputenc}
\usepackage[small]{caption}
\usepackage{graphicx}
\usepackage{amsmath}
\usepackage{amsthm}
\usepackage{booktabs}
\usepackage{algorithm}
\usepackage{algorithmic}
\usepackage[switch]{lineno}
\usepackage{amssymb}
\usepackage{multicol}
\usepackage{tasks}
\usepackage{booktabs}
\usepackage{enumitem}
\usepackage{cases}
\usepackage[nice]{nicefrac}

\newtheorem{thm}{Theorem}
\newtheorem*{thm*}{Theorem}
\newtheorem{lem}{Lemma}
\newtheorem*{lem*}{Lemma}

\newtheorem{Property}{Property}
\newtheorem*{Property*}{Property}

\title{MRCert: Towards Post-deployment Patch Robustness Certification for Adversarially Patched Samples via Type-specific Masking}

\author{
    Qilin Zhou$^{1}$, Zhengyuan Wei$^{2}$, Haipeng Wang$^{1}$, Zhuo Wang$^{1}$, Shuo Liu$^{1}$, W.K. Chan$^{1,*}$
    \affiliations
     $^{1}$City University of Hong Kong $^{2}$The Education University of Hong Kong\\
     $^{*}$Corresponding author
}

\begin{document}

\maketitle

\begin{abstract}
In post-deployment time, inputs to deep learning models may or may not be adversarially patched.
Patch robustness certification on such inputs within a patch bound can verify their label benignity and should retain high prediction accuracy.
However, existing smoothing-based and masking-based recovery  defenders cannot achieve both simultaneously: they degrade the prediction accuracy much
and cannot verify the benignity of the returned label of an adversarially patched input, respectively.
We propose MRCert, the first masking-based certified recovery defender that shows the feasibility of achieving both.
Unlike all existing works to apply a common condition across both types of input (benign and adversarially patched samples) for certification,
MRCert infers type-specific necessary properties of deep learning models for both types in post-deployment time and formally relates them to verify the label benignity through a novel type-oriented design of label recovery and certification function pair.
Without incurring the degradation in clean accuracy caused by smoothing, experimental results confirm that MRCert achieves 35.1\% adversarial certified accuracy on ImageNet at patch size 16 pixels, whereas the SOTA PatchCURE fails completely.

\end{abstract}

\section{Introduction}
Deep learning (DL) systems have achieved remarkable success in many application domains, including those safety-critical \cite{2021SurveyK,litjens2017survey};
yet their vulnerability to adversarial attacks \cite{szegedy2013intriguing}, including those physically realizable attacks like a sticker \cite{wei2023simultaneously} or a tag \cite{liu2020bias} in the real world, poses significant threats to their deployment, e.g., commercial autonomous vehicles \cite{wang2025revisiting}.
Patch adversarial attack \cite{brown2017adversarial,wei2024physical}
is a model of physically realizable attacks,
where attackers can arbitrarily manipulate pixels within a certain region of an image (referred to as a patch region) to compromise DL systems.

Deterministic certified defenders are emerging.
Among them, certified recovery against patch attacks \cite{levine2020randomized,metzen2021efficient,li2022vip,xiang2022patchcleanser,salman2022certified,zhou2023majority,zhou2024crosscert,xiang2024patchcure}, as a major line of patch robustness certification,
aims to ensure the recovery of benign labels from all patched neighbors of certified benign samples.
These existing  defenders are predominantly designed to
offer pre-deployment provable patch robustness for benign samples.
In practice, it is desirable to provide a certification on input samples
for certified recovery \emph{at post-deployment time},
especially when an input may already be corrupted,
where an issued certificate indicates that the label returned by the label recovery function for the input sample is robust and benign, so that downstream tasks can take informed actions.

However, current certified recovery  defenders, including both masking-based \cite{xiang2024patchcure,xiang2022patchcleanser} and smoothing-based methods
\cite{salman2022certified,li2022vip}, fall short in this regard. Masking-based ones rely on assumptions about the input's benignity and fail to certify any adversarially patched samples of certified benign samples (illustrated in  Fig.~\ref{fig:first_view} and formalized in \S~\ref{sec:moti}), while smoothing-based ones suffer from degraded correctness due to their smoothing mechanisms \cite{xiangshort}.

To address these limitations, we propose MRCert, the first certified recovery defender that provides a type-specific robustness certification crosscutting with its type-specific label recovery function, able to certify different types of input samples in post-deployment time.
MRCert can distinguish among benign, patched, and double-patched inputs with different recovery strategies on them in a type-specific manner.
It is able to not only recover the benign label for adversarially patched samples of certified benign samples (like its peers)
but also assign a provable certificate indicating whether the label returned by the recovery function for the input sample is robust and benign.
The type-specific design of MRCert enables it to certify both benign and adversarial samples, overcoming the fundamental limitations of peers.
In addition to maintaining a high certified accuracy and clean accuracy on benign samples, the experiment shows that MRCert achieves a non-trivial round-trip certified accuracy, which is further validated on patched samples from actual attacks.

The main contribution of this paper is twofold. (1) The paper proposes MRCert and proves its soundness.
(2) It validates MRCert's high effectiveness in certifying benign and adversarial samples while maintaining high accuracy.

\begin{figure}[t]
\centering
\includegraphics[width =\linewidth]{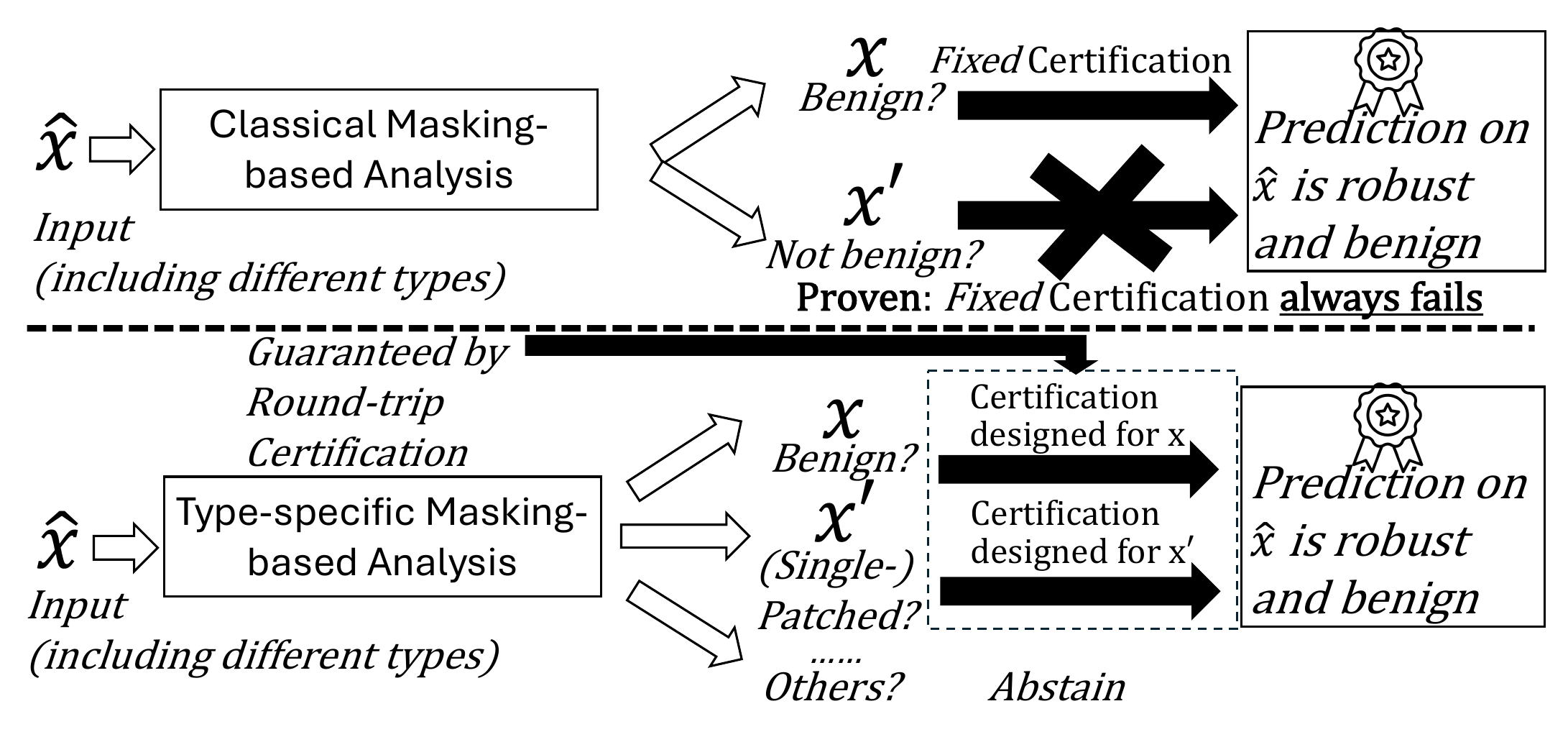}
\captionsetup{font={small}}
\caption{
Post-deployment certification in classical masking-based analysis \emph{vs} our type-specific masking-based analysis.
}
\label{fig:first_view}
\end{figure}

\section{Preliminaries}
We first summarize the main notation we use.
\setlength{\multicolsep}{2pt}
\begin{multicols}{2}
\begin{itemize}[nosep, leftmargin=*]
    \item $x$: a benign sample
    \item $x'$: $x$'s patched sample
    \item $\hat{x}$: an input sample
    \item $f$: \text{a base classification model}
    \item $c_r$: a certification function
    \item $g$: \text{a label recovery function}
    \nolinenumbers
    \item $\textsc{p}$: a patch region
    \item $\mathbb{P}$: a patch region set
    \item $\textsc{m}$: a mask
    \item $\hat{x}\odot\textsc{m}$: a mutant of $\hat{x}$
    \item $\mathbb{M}_\mathbb{P}$: a covering mask set
    \item $y_0$: the true label
    \item $\mathbb{A}_\mathbb{P}(.)$: all patched neighbors
\end{itemize}
\end{multicols}

\subsection{Patch attacks and Certified Recovery}
\subsubsection{Model and Patch}
An image classifer is a function $f: \mathcal{X} \rightarrow \mathcal{Y}$, which
takes $\hat{x}$ $\in \mathcal{X} \subset \mathbb{R}^{w \times h}$ (with its true class label $y_0$ $\in \mathcal{Y} = \{0, \cdots, |\mathcal{Y}|-1\}$) as input and produces a label $f(\hat{x}) \in \mathcal{Y}$.
A patch  in a sample is modeled to reside within a contiguous square \textbf{patch region}, represented by a binary matrix $\textsc{p} \in \mathbb{P} \subset \{0,1\}^{w \times h}$,
where 1s indicate the region and 0s elsewhere.
The set
$\mathbb{P}$ includes \textbf{all patch regions} satisfying predefined conditions (e.g., shape, size).
We use the attack constraint set $\mathbb{A}_\mathbb{P}({x})$ to define \textbf{all patched neighbors} (in terms of one patch by default) of a benign sample $x$:
$\mathbb{A}_\mathbb{P}({x})$ = $\{{x}' \mid {x}'=(\textsc{J}-\textsc{p}){\odot} {x}+\textsc{p}{\odot} {x}' \land \textsc{p}\in \mathbb{P} \}$.
Here, ${x}'$ is a patched neighbor of $x$ with possible pixel modification inside \textsc{p} ($\textsc{p}{\odot} {x}'$) while keeping those outside unmodified ($(\textsc{J}-\textsc{p}){\odot} {x}$).
 $\odot$ denotes element-wise multiplication, and \textsc{J} is an all-ones matrix.
Like prior certification works,
our theoretical model focuses on robustness but not correctness and assumes $f(x)=y_0$ when certifying $x$.
The sample ${x}'$ is \textbf{adversarially} patched if
 $f(x')\neq f(x)$.
 Note that a benign sample is a trivial (non-modified) patched neighbor of itself.

 In this work, we build our theory on top of the assumption like \cite{yatsura2023certified} that the post-deployment input domain consists of benign samples and their patched neighbors.

\subsubsection{Certified Recovery}

A certified recovery  defender $R = \langle g, c_r\rangle$ is built on a base classifier $f$. The label recovery function $g$ provably returns the prediction of a benign sample $x$ for any patched neighbor $x'$ of $x$, i.e., $g(x'){=}g(x)$ if $x$ is \textbf{certified} by $c_r$, i.e.,  $[\forall x'{\in}\mathbb{A}_\mathbb{P}(x), c_r(x)$ ${=}\textit{True}\,{\implies}\,g(x'){=}g(x)]$.
Given a benign sample $x$, an attacker aims to find a patched neighbor $x'$, such that $g(x')\neq g(x)$.

\subsection{Post-deployment Robustness Certificate}
\cite{yatsura2023certified} proposes a post-deployment robustness certificate for certified patch recovery:
Since
$x \in \mathbb{A}_\mathbb{P}(\hat{x})$, one must have $g(\hat{x}){=}g({x})$ if $c_r(\hat{x}){=}\textit{True}$ (i.e., $[\forall \hat{x}' \in \mathbb{A}_\mathbb{P}(\hat{x}), g(\hat{x}'){=}g(\hat{x})]$) holds, indicating that the prediction $g(\hat{x})$ must be robust and further benign.
\textbf{Problem:}
Existing  defenders
are either inapplicable to provide such a certificate whenever inputs are adversarially patched samples of certified samples \cite{xiang2024patchcure} or able to provide but less accurate for benign samples \cite{salman2022certified} (see \S \ref{sec:moti}), limiting their applicability to provide the certificate in post-deployment time.

\subsection{Masking-Based Certified Recovery}
\label{sec:patchclenaser}

A mask $\textsc{m} \in \mathbb{M}_\mathbb{P} \subset \{0,1\}^{w \times h}$ ablates a specific region from a sample (0 for mask, 1 otherwise). The covering mask set $\mathbb{M}_\mathbb{P}$ \cite{xiang2024patchcure} ensures any patch $\textsc{p} \in \mathbb{P}$ is fully enclosed by some mask $\textsc{m} \in \mathbb{M}_\mathbb{P}$, i.e., $\forall \textsc{p} \in \mathbb{P}, \exists \textsc{m} \in \mathbb{M}_\mathbb{P}, \textsc{p} {\odot} \textsc{m} = \mathbf{O}$ (an all-zero matrix).
Ablating $\hat{x}$ with $k \in \{1,\dots,4\}$ masks from $\mathbb{M}_\mathbb{P}$ produces a $k$-th order \textbf{mutant}, denoted as $\hat{x} {\odot} m_1 {\odot} \dots {\odot} m_k$ (Fig.~\ref{fig:mask-overview_concept} for $k$=1).
Re-applying the same mask to a mutant trivially yields a higher-order one.

A masking-based certified detection defender \cite{patchcensor} raises a warning if any first-order mutant of an input $\hat{x}$ deviates from its original prediction (called \textbf{divergent} mutant), i.e., $\exists \textsc{m} \in \mathbb{M}_\mathbb{P}$ such that $f(\hat{x} {\odot} \textsc{m}) \neq f(\hat{x})$.

Extending this, a certified recovery defender $R = \langle g, c_r\rangle$ seeks to recover the benign prediction of $x'$
\cite{xiang2024patchcure}: $g$ generates all second-order mutants of $x$ for such divergent mutants to test them and returns the prediction of a divergent mutant if all of them and the divergent mutant yield the same prediction.
If no such divergent mutant is identified at either stage, $g$ returns $f(\hat{x})$.
This mechanism is formalized by the following theorem:
$[\forall \textsc{m}_1, \textsc{m}_2\in\mathbb{M}_\mathbb{P}, f(x\odot\textsc{m}_1\odot\textsc{m}_2)=y_0] \implies [\forall x'\in\mathbb{A}_\mathbb{P}(x), \exists \textsc{m}_1 \in\mathbb{M}_\mathbb{P}, \forall \textsc{m}_2 \in\mathbb{M}_\mathbb{P}, f(x'{\odot} \textsc{m}_1{\odot} \textsc{m}_2)=y_0]$, where the antecedent serves as the certification function, and the consequent yields prediction as described above.

\begin{figure*}[t]
\centering
\includegraphics[width =0.97\linewidth]{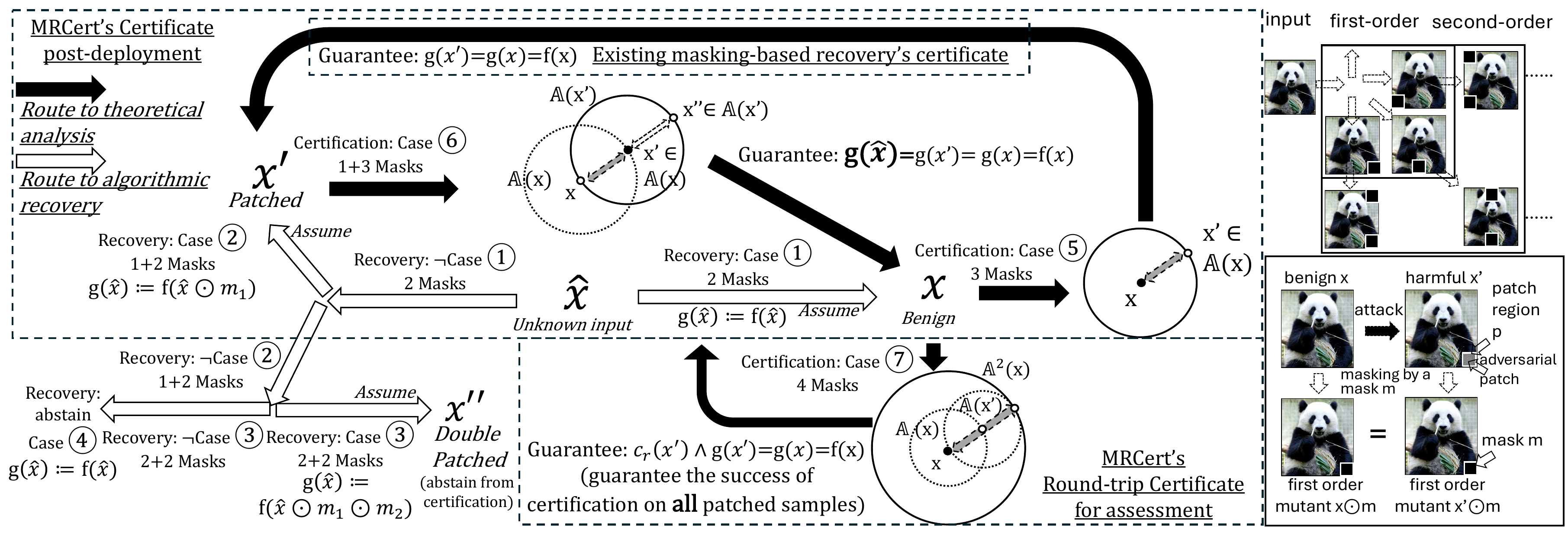}
\captionsetup{font={small}}
\caption{The basic concepts in masking-based defenders (right) and the illustration of MRCert (left).
See Fig.~\ref{fig:MRCert_formulation} for their formulation.
}

\label{fig:mask-overview_concept}
\end{figure*}

\section{Motivation and Challenges}\label{sec:moti}

Post-deployment robustness certification for benign samples and their patched samples is highly desirable.
Such a formal guarantee empowers users to substantially augment automation in safety-critical scenarios for those certified by establishing a provable basis for trust.
One needs a defender to certify input samples with uncertainty about their benignity while maintaining high clean accuracy.
However, both smoothing-based and masking-based certified recovery have their own limitations and challenges in achieving the goal.

\textit{Masking-based Recovery.}
To handle inputs with unknown groundtruth $y_0$ at post-deployment time, we follow \cite{patchcensor} to replace $y_0$ with $f(x)$ in the certification in \cite{xiang2022patchcleanser}:
$c_r(\hat{x}) := [\forall \textsc{m}_1, \textsc{m}_2 \in \mathbb{M}_\mathbb{P}, f(\hat{x} {\odot} \textsc{m}_1 {\odot} \textsc{m}_2) = f(\hat{x})]$.
However, we find that \emph{no} adversarially patched sample of a certified sample $x'$ by this $c_r$ can be further certified.

\begin{Property}[Infeasibility of Re-certifying Adversarially Patched Samples for Existing Masking-based Recovery Works]
$
\forall x' \in \mathbb{A}_{\mathbb{P}}(x), c_r(x) \land f(x') \neq f(x) \implies \neg c_r(x').
    $
\end{Property}

\textit{Idea of proof.}
If $c_r(x)$ is satisfied, all its first-order mutants predict $f(x)$. Since $x'$ shares a first-order mutant with $x$ but $f(x') \neq f(x)$, it cannot satisfy the prediction consensus required for $c_r(x')$.
See the formal proof in the appendix.

\textit{Smoothing-based recovery.}
Smoothing is known to seriously degrade clean accuracy on non-trivial datasets like ImageNet \cite{mohapatra2021hidden}, including those for patch attacks \cite{xiangshort}.
Previous works have proposed an \textit{empirical} masking-based detection method atop smoothing, but the improvement in clean accuracy is only marginal (e.g., 0.8\% on CIFAR10 in \cite{xiang2021patchguard}).

\section{Methodology}

\subsection{Overview of MRCert}
Fig.\ref{fig:mask-overview_concept} depicts an overview of
our certified recovery defender, MRCert $R = \langle g, c_r, c_r^2\rangle$, formulated in Fig.~\ref{fig:MRCert_formulation}.
The label recovery function $g(\hat{x})$ returns a label structured as a nested-if sequence: \texttt{if \textcircled{1} elseif \textcircled{2} elseif \textcircled{3} else \textcircled{4}}.
The certification function $c_r(\hat{x})$ (cases \textcircled{5}--\textcircled{6}) verifies if $g$'s output remains invariant for $\hat{x}$ and all neighbors in $\mathbb{A}_\mathbb{P}(\hat{x})$.

Our setting involves patched neighbors not present in the clean dataset. Since computing the full set of patched neighbors is intractable, we provide provable assessments that map to the clean dataset, formulated as $c^2_r$ in case \textcircled{7}.

At post-deployment, MRCert is configured with a base model $f$, accepts an input $\hat{x}$, and applies $g$ and $c_r$ on the input to return a label and a certification verdict on robustness.

MRCert treats benign samples (including those benignly patched) and their adversarially patched neighbors as distinct input types.
It infers the type-specific necessary properties of $f$ for these two types that can be detected, recovered, and certified, and links up these properties through a type-oriented design:
For benign samples, it detects through case \textcircled{1}, certifies through  case \textcircled{5} if case \textcircled{1}  holds;
for adversarially patched neighbors of benign samples,
it detects through the violations of \textcircled{1} and further detects and recovers through  case \textcircled{2}, and certifies through  case \textcircled{6}.
In the theoretical model of an input $\hat{x}$ being either benign or patched, even if the input is adversarially patched, $c_r(\hat{x})$ can hold true to indicate that $g(\hat{x})$ must return a benign label.
Existing masking-based defenders cannot offer this property.
Also,
if case \textcircled{7} holds, then
all the conditions for the above certification, detection, and recovery will be guaranteed to hold, as depicted by the solid arrows in Fig.~\ref{fig:mask-overview_concept} (and
the hollow arrows depict cases \textcircled{1}--\textcircled{4} of $g$).
Swapping the order of cases \textcircled{1}--\textcircled{4} in $g$ will make case \textcircled{7} for round-trip certification fail (see the appendix for discussion).

\begin{figure*}\centering
\begin{minipage}{\textwidth}

\begin{numcases}{g(\hat{x}):=}
f(\hat{x}),
& if $\forall {\textsc{m}_1}, {\textsc{m}_2}\in\mathbb{M}_\mathbb{P}, f(\hat{x}\odot\textsc{m}_1\odot\textsc{m}_2)= f(\hat{x})$
\tag{case \textcircled{1}}
\label{g-case-detecton-of-patch-existence}
\\
f(\hat{x}\odot{\textsc{m}_1}),
& else if  $\exists {\textsc{m}_1}\in\mathbb{M}_\mathbb{P},\forall {\textsc{m}_2}, {\textsc{m}_3}\in\mathbb{M}_\mathbb{P}, f(\hat{x}\odot\textsc{m}_1\odot\textsc{m}_2\odot\textsc{m}_3)= f(\hat{x}\odot{\textsc{m}_1})$
\tag{case \textcircled{2}}
\label{g-case-identify-patched-only}
\\
f(\hat{x}\odot\textsc{m}_1\odot\textsc{m}_2),
& \parbox[t]{0.65\textwidth}{
  else if $\exists {\textsc{m}_1},{\textsc{m}_2} (\neq {\textsc{m}_1})\in\mathbb{M}_\mathbb{P},\forall {\textsc{m}_3}, {\textsc{m}_4}\in\mathbb{M}$, \\
  $f(\hat{x}\odot{\textsc{m}_1}\odot{\textsc{m}_2}$
  $\odot{\textsc{m}_3}\odot{\textsc{m}_4})=$
  $f(\hat{x}\odot\textsc{m}_1\odot\textsc{m}_2)$
}
\tag{case \textcircled{3}}
\label{g-case-empricially-recover-if-double-patchedes}
\\
f(\hat{x}),
& \text{otherwise}
\tag{case \textcircled{4}}
\label{g-case-fallback-strategy-no-recovery}
\end{numcases}
\vspace{-1.5ex}
\begin{empheq}[left={\hspace*{-1cm}\text{\framebox{certificate}} \ \  c_r(\hat{x}):= \empheqlbrace}]{align*}
  & [\forall \text{\textsc{m}}_1, \text{\textsc{m}}_2, \text{\textsc{m}}_3 \in \mathbb{M}_\mathbb{P}, f(\hat{x} {\odot} \text{\textsc{m}}_1 {\odot} \text{\textsc{m}}_2 {\odot} \text{\textsc{m}}_3) = f(\hat{x})]
  & \text{if case \textcircled{1} is taken in } g \tag{case \textcircled{5}} \\
  & [\exists \text{\textsc{m}}_1 \in \mathbb{M}_\mathbb{P}, \forall \text{\textsc{m}}_2,\text{\textsc{m}}_3,\text{\textsc{m}}_4 \in \mathbb{M}_\mathbb{P}, \\
  & f(\hat{x} {\odot} \text{\textsc{m}}_1 {\odot} \text{\textsc{m}}_2 {\odot} \text{\textsc{m}}_3 {\odot}
  \text{\textsc{m}}_4) = f(\hat{x} {\odot} \text{\textsc{m}}_1)]
  & \text{if case \textcircled{2} is taken in } g \tag{case \textcircled{6}}
\end{empheq}
\begin{numcases}{\framebox{round-trip certificate} \ \ \ c_r^2(\hat{x}):=}
[\forall {\textsc{m}_1},{\textsc{m}_2},{\textsc{m}_3},{\textsc{m}_4}\in\mathbb{M}_\mathbb{P}, f(\hat{x}\odot{\textsc{m}_1\odot\textsc{m}_2\odot\textsc{m}_3\odot\textsc{m}_4})=f(\hat{x})]
 \tag{case \textcircled{7}}
  \label{cr-case-for-neighbor-of-patched-sample}
  \label{Eq: double_certification}
\end{numcases}

\end{minipage}
\caption{
The formulation of each component of MRCert.
}
\label{fig:MRCert_formulation}
\end{figure*}

\subsection{Needs in Certifying Adversarial Samples}
\label{sec:needs}
In the theoretical model of certified recovery \emph{subject to post-deployment certification},
the input is either benign $x$ or patched $x'$. To certify the robustness of an adversarially patched neighbor $x'$ of $x$, the certification analyzer must analyze all neighbors in $\mathbb{A}_\mathbb{P}(x')$, which may contain up to two patches relative to $x$, including $x$ itself (clean), $x'$ (patched), and double-patched instances.
The fundamental limitation of existing masking-based peers \cite{xiang2022patchcleanser,xiang2024patchcure} stems from using the \emph{fixed-masking paradigm}:
By design, their $g$ and $c_r$ always apply a fixed number of masks to each input regardless of their patch count, forcing them to test the consistency of the predictions from the input's mutants under the same condition formulation.
To defend two patches, such a defender must scale its number of masks: (a) $g$ must generate second-order mutants for detection; (b) recovery requires fourth-order mutants to resolve a second-order divergent mutant.
This change
can only extend the guarantee for $x$ against the expanded set $\mathbb{A}_\mathbb{P}^2(x) = \bigcup_{x' \in \mathbb{A}_\mathbb{P}(x)} \mathbb{A}_\mathbb{P}(x')$, and remains unable to certify $x' \in \mathbb{A}_\mathbb{P}(x)$.

Their core bottleneck is over-approximation: To recover $x'$, their common theoretical model treats every input and its patched neighbors as double-patched instances; therefore, to certify $x'$, it must be further modified into yet another Patchcleanser variant to defend against four patches, creating an unresolvable ``cat-and-mouse" recursion. So,
any fixed-mask configuration is fundamentally unable to certify $x'$.

\subsection{Type-Specific Label Recovery Function
}\label{sec:self-adaptive}

The primary challenge in designing $g$ lies in the agnostic nature of the input types:
the mechanism must certify both types of input, $x$ and $x' \in \mathbb{A}_\mathbb{P}(x)$, without prior knowledge of the patch count in an input. Under the above fixed-masking paradigm, certifying a benign sample $x$ requires the label consistency across all its fourth-order mutants. Let us first adopt this property for ${x}$ to discuss the design of our $g$.

Our $g$ follows a conditional pipeline to resolve issues arising from different input types
as discussed in \S\ref{sec:needs}. In Case \textcircled{1}, $g$ screens $\hat{x}$ using second-order mutants to detect adversarial patches, ensuring that double-patched samples cannot evade the detection.
Then, the core question is:
\textbf{How should $g$ be designed to discriminate the number of adversarial patches in $\hat{x}$ after the above detection procedure?}

Our \textbf{first insight} is that testing all first-order mutants of $\hat{x}$ by the prediction consistency of the third-order mutants of $\hat{x}$ not only exposes which of these first-order mutants that $f$ returns the benign label, but also indicates whether $\hat{x}$ contains a single (adversarial) patch.
We refer to it as \textit{indirect testing} of adversarially patched samples.
The elegance of this test lies in its dual roles: while identifying and recovering the labels for inputs with one patch, at the same time, it detects inputs that have more than one patch by the consistency of third-order mutants,
thereby \emph{precluding any evasion attempts} of double-patched neighbors
if $g$ naively uses two masks (just like \cite{xiang2022patchcleanser}) for inputs with one patch to recover its label, since the first mask may miss the patch.
So, $g$ identifies $\hat{x}$ to be $x' \in \mathbb{A}_\mathbb{P}(x)$ based on this indirect testing, enabling certified recovery in Case \textcircled{2}.
It solves the problem of testing $\hat{x}$ that directly uses a test suite consisting of either the set of all first-order mutants or the set of all second-order mutants,
which either fails to ensure the robustness guarantee on $\hat{x}$,
or is too coarse in detection granularity.
To our knowledge, MRCert is the \textit{first work} that is able to detect the number of patches on an adversarial input sample with a guarantee
if its benign counterpart is round-trip certified (see \S\ref{sec:certification}).

Thus, if $\hat{x}$ is detected as an adversarially patched sample in Case \textcircled{1} but not detected as a sample with one patch in Case \textcircled{2}, $g$ deems it as a sample with two (or more) adversarial patches.
MRCert abstains from verifying its robustness, and $g$ adopts the existing recovery on double-patched samples in Case \textcircled{3} as a best effort attempt.
If all these test efforts cannot find a label to return, Case \textcircled{4}
returns
the original prediction.

We will present the formal guarantee in \S \ref{sec:certification}.
The appendix provides
further justifications for our design choices, including the necessity of the current sequence and mask quantity.

\subsection{Post-deployment Type-specific Certification and Round-trip Certification for Assessment}\label{sec:certification}
Building upon the proposed type-oriented recovery mechanism, we formalize the robustness certification conditions for MRCert against both benign samples and adversarially patched neighbors.

First, to ensure that an adversarial sample $x'$ (with one patch) is caught by Case \textcircled{1} of $g$, we must guarantee the existence of at least one second-order mutant of $x'$ whose predicted label differs from that of $x'$.
By definition, we know $f(x')\neq f(x)$.
If $x$ and all its first-order mutants share the same prediction, then the predictions for first-order mutant of $x'$ that masks the adversarial patch will revert to $f(x)$, further differing from $f(x')$, which is detected by Case \textcircled{1}.
Formally, we require $[\forall \textsc{m}_1 \in \mathbb{M}_\mathbb{P}, f(x{\odot}\textsc{m}_1) = f(x)] {\implies} [\forall x' \in \mathbb{A}_\mathbb{P}(x), f(x') \neq f(x) {\implies} \exists \textsc{m}_1 \in \mathbb{M}_\mathbb{P}, f(x'{\odot}\textsc{m}_1) \neq f(x')]$ holds.
Next, to recover the benign labels for all $x' \in \mathbb{A}_\mathbb{P}(x)$), we must ensure that they are correctly identified under Case \textcircled{2} of $g$. Since the specific mask $\textsc{m}_1$ required to neutralize the adversarial patch is unknown a priori, we formulate a consistency requirement:
for every first-order mutant, the entire set of its corresponding third-order mutants must yield a prediction identical to that of the first-order mutant itself.
Formally, we require $[\forall \textsc{m}_1, \textsc{m}_2, \textsc{m}_3 \in \mathbb{M}_\mathbb{P}, f(x{\odot}\textsc{m}_1{\odot} \textsc{m}_2{\odot}\textsc{m}_3) {=} f(x{\odot}\textsc{m}_1)] {\implies} [\exists \textsc{m}_1 {\in} \mathbb{M}_\mathbb{P}, \forall \textsc{m}_2, \textsc{m}_3$ ${\in} \mathbb{M}_\mathbb{P}, f(x'{\odot}\textsc{m}_1{\odot}\textsc{m}_2 {\odot}\textsc{m}_3) {=} f(x' {\odot} \textsc{m}_1)]$ holds.
A certified sample is required to satisfy the antecedents of both, ensuring being detected in Case \textcircled{1} and recovered in Case \textcircled{2} for all of its adversarially patched samples.
By substituting the input as $\hat{x}$, we derive: $[\forall \textsc{m}_1, \textsc{m}_2, \textsc{m}_3 \in \mathbb{M}_\mathbb{P},$ $ f(\hat{x} {\odot} \textsc{m}_1 {\odot} \textsc{m}_2 {\odot} \textsc{m}_3) {=} f(\hat{x})]$. This is Case \textcircled{5}, which we formally define it in Lemma~\ref{lem:recovery_benign}.
The robustness and benignity of an input sample in Case \textcircled{1} of $g$ can be verified through it.

\begin{lem}[Certification of benign samples]
\label{lem:recovery_benign}
Given an arbitrary sample $\hat{x}$,
if $[\forall {\textsc{m}_1},{\textsc{m}_2},{\textsc{m}_3}\in\mathbb{M}_\mathbb{P}, f(\hat{x}\odot{\textsc{m}_1}\odot{\textsc{m}_2}\odot{\textsc{m}_3})= f(\hat{x})]$ holds,
$\forall \hat{x}'\in\mathbb{A}_\mathbb{P}(\hat{x}),
g(\hat{x}')=g(\hat{x})
$.
\end{lem}

Our \textbf{second insight} is that the certification function should also be \emph{type-specific and crosscutting the label recovery function on a type-specific basis, rather than type-agnostic and acting purely as a downstream task of the label recovery function.}
In previous patch robustness certification paradigms, a certification function is agnostic in input type (see
\S\ref{sec:patchclenaser}) and relates to the output of a label recovery function without knowing the execution path condition that produces the output.
However, it is not optimal for masking-based defenders to handle the situation where benign samples and their adversarially patched samples are subject to robustness certification (see \S \ref{sec:moti}), which requires different treatments (see \S \ref{sec:needs}).
Therefore, we design our certification and label recovery function pair to correspond to each other in a type-specific manner.
To certify a adversarially patched sample as input,
we need to ensure that the returned label remains the same across the input and all its patched neighbors, where double-patched samples should be included.
Samples with two adversarial patches
are expected to be detected by violating the condition of Case \textcircled{2} and their labels are recovered in Case \textcircled{3} in the design of $g$ presented in \S \ref{sec:self-adaptive}.
To ensure achieving such a goal, $x'$ should have a first-order mutant (denoted as $x'{\odot}\textsc{m}_1$), whose all fourth-order mutants are predicted with the same prediction
$f(x'{\odot}\textsc{m}_1)$ (which should be the same as $g(x')$, the return prediction of $x'$, by the conditions in Case \textcircled{1} and Case \textcircled{2}).
With such a condition, for all double-patched samples (denoted as $x''$) of this adversarially patched sample (i.e., putting another patch on this adversarially patched sample),
there should have a second-order mutant of $x''$ (denoted as $x''{\odot}\textsc{m}_1{\odot}\textsc{m}_2$) generated from this first-order mutant $x''{\odot}\textsc{m}_1$(=$x'{\odot}\textsc{m}_1$) with the second mask $\textsc{m}_2$ covering the new patch in the input, and $f$ predicts the same label for this second-order mutant $x''{\odot}\textsc{m}_1\odot\textsc{m}_2$ and all fourth-order (and third-order) mutants generated from this second-order mutant.
So, all adversarially double-patched samples $x''$ (i.e., whose attacker aims to make $g(x'')\neq g(x')$ where $x''{\in}\mathbb{A}_\mathbb{P}(x')$)
should meet the following three conditions:
(1) $x''$
must violate the condition in Case \textcircled{1} because at least one second-order mutant $x''{\odot}\textsc{m}_1{\odot}\textsc{m}_2$ predicts to $g(x')$;
(2) $x''$
must violate the condition in Case \textcircled{2} because all third-order mutants of $x''{\odot}\textsc{m}_1{\odot}\textsc{m}_2$ predicts to $g(x')$; and,
(3) $x''$ must satisfy the condition in Case \textcircled{3} by $x''{\odot}\textsc{m}_1{\odot}\textsc{m}_2$ and all fourth-order (and third-order) mutants generated from this second-order mutant predicted to $g(x')$, and further output the label of $x''{\odot}\textsc{m}_1{\odot}\textsc{m}_2$, which is $g(x')$.
Similarly, by substituting the symbol $\hat{x}$ for $x'$, it is captured as Lemma~\ref{lem:recovery_one_patched_sample} and as the condition of Case \textcircled{6}, where the robustness and benignity of an input sample in Case \textcircled{2} of $g$ can be verified.

\begin{lem}[Certification of adversarially patched samples]
\label{lem:recovery_one_patched_sample}
Given an sample $\hat{x}$,
if $[\exists {\textsc{m}_1}\in\mathbb{M}_\mathbb{P},\forall$ ${\textsc{m}_2}$,${\textsc{m}_3}$,${\textsc{m}_4}\in\mathbb{M}_\mathbb{P}, f(\hat{x}\odot\textsc{m}_1\odot\textsc{m}_2\odot\textsc{m}_3\odot\textsc{m}_4)= f(\hat{x}\odot{\textsc{m}_1})]$ holds,
$[\forall \hat{x}'\in\mathbb{A}_\mathbb{P}(\hat{x}),
g(\hat{x}')=g(\hat{x})
]$ holds.
\end{lem}

In short, the labels for a sample $x$ of the benign type and its patched neighbors are recovered through Case \textcircled{1} and either Case \textcircled{1} (for benign patch) or \textcircled{2} (for adversarial patch), respectively. Similarly, these for an adversarial sample with one patch and its patched neighbors are recovered through Cases \textcircled{2} and Cases \textcircled{1}--\textcircled{3}, respectively.

With these two conditions (shown as cases \textcircled{5} and \textcircled{6} in Fig. \ref{fig:MRCert_formulation}), we get MRCert's certification function by Thm.~\ref{thm:recovery}.

\begin{thm}[Certification of samples]
\label{thm:recovery}
Given an arbitrary sample $\hat{x}$,
if $[\forall {\textsc{m}_1},{\textsc{m}_2},{\textsc{m}_3}\in\mathbb{M}_\mathbb{P}, f(\hat{x}{\odot}\textsc{m}_1{\odot}\textsc{m}_2{\odot}\textsc{m}_3)$ $= f(\hat{x})] \lor
[\exists {\textsc{m}_1}{\in}\mathbb{M}_\mathbb{P},\forall {\textsc{m}_2}, {\textsc{m}_3}, {\textsc{m}_4}\in\mathbb{M}_\mathbb{P}, f(\hat{x}{\odot}\textsc{m}_1{\odot}\textsc{m}_2{\odot}\textsc{m}_3{\odot}\textsc{m}_4){=}f(\hat{x}$ ${\odot}{\textsc{m}_1})]$ holds (i.e., $c_r(\hat{x})=\textit{True}$),
$\forall \hat{x}'{\in}\mathbb{A}_\mathbb{P}(\hat{x}),
g(\hat{x}')=g(\hat{x})
$.
\end{thm}

If a benign sample and all its patched samples can be certified  \cite{yatsura2023certified}, it is called \textbf{round-trip certified}.

Suppose $c_r$ certifies a benign sample $x$ itself.
Then if all its patched samples are guaranteed to satisfy the antecedent of Thm.~\ref{thm:recovery} to be certified, $x$ is round-trip certified.
For masking-based defenders,
since it is unknown which mask will cover the patch in a patched sample created from a benign sample, to make the patched sample fulfill the antecedent of Thm.\ref{thm:recovery} that it has a first-order mutant, whose all four-order mutants generated from this mutant are predicted with the same label as this first-order mutant, we also need to make sure that all first-order mutants of the benign sample $x$ fulfill this antecedent, which is captured by Thm.~\ref{thm:recovery_double}.

\begin{thm}[Round-trip certification of benign samples]
\label{thm:recovery_double}
Given a benign sample ${x}$, if
$[\forall {\textsc{m}_1},{\textsc{m}_2},{\textsc{m}_3},{\textsc{m}_4}\in\mathbb{M}_\mathbb{P}, f({x}\odot$
$\textsc{m}_1$$\odot$$\textsc{m}_2$$\odot\textsc{m}_3$$\odot\textsc{m}_4)=$
$f(x)]$ holds
(i.e., $c_r^2(x)=\textit{True}$),
$[\forall {x}'
\in
\mathbb{A}_\mathbb{P}
$
$({x}),
g({x}')$
$
=
g({x})
\land
c_r({x}')=\textit{True}\land c_r({x})=\textit{True}]
$ holds.
\end{thm}

$c_r^2$ adopts the antecedent of the round-trip certification theorem (case \textcircled{7} in Fig. \ref{fig:MRCert_formulation}), which ensures $x$ and all its patched samples fulfill Thm.~\ref{thm:recovery}.
It aligns with the certification condition against two patches
by the Patchcleanser variant in \S\ref{sec:moti}.

Formal proofs and the extension of MRCert to handle up to $N$ patches are provided in the appendix.

\section{Evaluation}

\begin{table}[]
\centering
\captionsetup{font={small}}
\caption{Results of smoothing-based recovery defenders on clean ImageNet test set with patch size 32 pixels (in \%).}\label{tab:smooth_results}
{
\begin{tabular}{|c|c|c|}
\hline
Defender & $acc_{clean}$ & $acc_{cert}$ \\ \hline
S-ViT \cite{salman2022certified} & 69.3 & 38.3 \\ \hline
ViP \cite{li2022vip} & 70.8 & 41.4 \\ \hline
VOT \cite{zhou2024crosscert} & 71.4 & 41.2 \\ \hline
\end{tabular}
}
\end{table}

\subsection{Experimental Setup}
\subsubsection{Environment and Dataset}
The evaluation is run on an Ubuntu 20.04 machine equipped with four Nvidia 3090 GPUs.
Following \cite{xiang2022patchcleanser}, we
adopt 1000-class Image-Net \cite{deng2009imagenet}, 10-class CIFAR10 \cite{krizhevsky2009learning}, and 10-class ImageNette as our datasets.
All samples are resized to $224$X$224$ pixels.

\begin{table*}[]
\caption{Results
on
$\mathbb{D}_{pat}$ (patched samples from a real attack) with patch sizes of 16 pixels and 32 pixels (in \%). }\label{tab:attack}
\resizebox{\linewidth}{!}{
\begin{tabular}{|c||ccccccc||ccccccc|}
\hline
Patch   size & \multicolumn{7}{c||}{16   pixels}                                                                                                                                            & \multicolumn{7}{c|}{32 pixels}                                                                                                                                             \\ \hline
Defender    & \multicolumn{3}{c|}{MRCert}                                                       & \multicolumn{3}{c|}{PC}                                                          & Base & \multicolumn{3}{c|}{MRCert}                                                       & \multicolumn{3}{c|}{PC}                                                         & Base \\ \hline
Metrics      & \multicolumn{1}{c|}{$acc_{clean}$}  & \multicolumn{1}{c|}{$acc_{cert}$}  & \multicolumn{1}{c|}{$acc^{adv}_{cert}$}  & \multicolumn{1}{c|}{$acc_{clean}$}  & \multicolumn{1}{c|}{$acc_{cert}$}  & \multicolumn{1}{c|}{$acc^{adv}_{cert}$} & $acc_{clean}$  & \multicolumn{1}{c|}{$acc_{clean}$}  & \multicolumn{1}{c|}{$acc_{cert}$}  & \multicolumn{1}{c|}{$acc^{adv}_{cert}$}  & \multicolumn{1}{c|}{$acc_{clean}$}  & \multicolumn{1}{c|}{$acc_{cert}$} & \multicolumn{1}{c|}{$acc^{adv}_{cert}$} & $acc_{clean}$  \\ \hline
ImageNet     & \multicolumn{1}{c|}{62.4} & \multicolumn{1}{c|}{44.2} & \multicolumn{1}{c|}{35.1} & \multicolumn{1}{c|}{67.6} & \multicolumn{1}{c|}{12.8} & \multicolumn{1}{c|}{0.0} & 14.6 & \multicolumn{1}{c|}{49.0} & \multicolumn{1}{c|}{24.8} & \multicolumn{1}{c|}{24.8} & \multicolumn{1}{c|}{62.2} & \multicolumn{1}{c|}{0.0} & \multicolumn{1}{c|}{0.0} & 0.0  \\ \hline
ImageNette   & \multicolumn{1}{c|}{97.4} & \multicolumn{1}{c|}{93.8} & \multicolumn{1}{c|}{88.1} & \multicolumn{1}{c|}{98.4} & \multicolumn{1}{c|}{42.0} & \multicolumn{1}{c|}{0.0} & 54.6 & \multicolumn{1}{c|}{95.2} & \multicolumn{1}{c|}{80.2} & \multicolumn{1}{c|}{80.2} & \multicolumn{1}{c|}{98.0} & \multicolumn{1}{c|}{0.0} & \multicolumn{1}{c|}{0.0} & 0.0  \\ \hline
CIFAR10      & \multicolumn{1}{c|}{93.0} & \multicolumn{1}{c|}{77.4} & \multicolumn{1}{c|}{41.5} & \multicolumn{1}{c|}{94.4} & \multicolumn{1}{c|}{61.0} & \multicolumn{1}{c|}{0.0} & 67.2 & \multicolumn{1}{c|}{82.4} & \multicolumn{1}{c|}{46.8} & \multicolumn{1}{c|}{46.6} & \multicolumn{1}{c|}{91.0} & \multicolumn{1}{c|}{0.2} & \multicolumn{1}{c|}{0.0} & 0.8  \\ \hline
\end{tabular}
}
\end{table*}

\subsubsection{Models and Baselines}
We adopt the Vision Transformer (ViT) as our base model, which achieves the state-of-the-art in many patch robustness certification defenders \cite{xiang2024patchcure,li2022vip,salman2022certified,xiang2022patchcleanser,patchcensor}.

We adopt the state-of-the-art masking-based defender
PatchCURE (\textbf{PC}) \cite{xiang2024patchcure} and state-of-the-art smoothing-based defender
VOT-CrossCert (\textbf{VOT}) from CrossCert \cite{zhou2024crosscert}.
We arm VOT with the methodology proposed by \cite{yatsura2023certified} for post-deployment certification for smoothing-based defenders
to compare with \textbf{MRCert} in our infrastructure.
They share the same PC's ViT-SRF as the base model.
We extract the results of peer smoothing-based recovery defenders from \cite{li2022vip} and \cite{salman2022certified} to compare with VOT in Table~\ref{tab:smooth_results}.

\subsubsection{Metrics}
\label{sec:metrics}
Let ${x}$ be a benign sample with the true label $y_0$ in a clean test dataset $\mathbb{D}$, and $R = \langle g(x), c(x),  c^2_r(x)\rangle$ be a certified recovery defender.
We adopt \textbf{clean accuracy}
$acc_{clean}=\nicefrac{\mid\{{x}\in\mathbb{D}\mid g({x})=y_0\}\mid}{\mid\mathbb{D}\mid}$, which evaluates the standard performance of certified recovery as a classifier,
\textbf{certified accuracy}
$acc_{cert}=\nicefrac{\mid\{{x}\in\mathbb{D}\mid g({x})=y_0\land c_r(x)=\textit{True}\}\mid}{\mid\mathbb{D}\mid}$, and
\textbf{round-trip certified accuracy}
$acc_{cert^2}=\nicefrac{\mid\{{x}\in\mathbb{D}\mid g({x})=y_0\land c_r^2(x)=\textit{True}\}\mid}{\mid\mathbb{D}\mid}$ as our metrics.
The clean accuracy of
the base classification model $f$ is
$\nicefrac{\mid\{{x}\in\mathbb{D}\mid f({x})=y_0\}\mid}{\mid\mathbb{D}\mid}$.
Existing works commonly adopt the first two metrics,
and round-trip certified accuracy is proposed in \cite{yatsura2023certified}.
For RQ1, $\mathbb{D}$ is a set of patched samples (including benignly and adversarially patched, one per benign) obtained by attacking 500 random benign samples using the procedure adopted from  \cite{levine2020randomized}, denoted as $\mathbb{D}_{pat}$.
We  measure each defender's ability to certify adversarially patched samples by  \textbf{adversarial certified accuracy}
$acc_{cert}^{adv}=\nicefrac{\mid\{\hat{x}\in\mathbb{D}_{pat}\mid g(\hat{x})= y_0\land c_r(\hat{x})=\textit{True} \land f(\hat{x})\neq y_0\}\mid}{\mid\{\hat{x}\in\mathbb{D}_{pat}\mid f(\hat{x})\neq y_0\}\mid}$.

The appendix presents more details on our setup.

\begin{table}[t]
\centering
\captionsetup{font={small}}
\caption{Results on ImageNet benign samples with patch size 16 (upper) and 32 (bottom) pixels (in \%).}\label{tab:benign_imagenet}

\begin{tabular}{|c|ccc|}
\hline
Dataset      & \multicolumn{3}{c|}{ImageNet}                                                        \\ \hline
Patch Size   & \multicolumn{3}{c|}{16 pixels}                                                            \\ \hline
Metrics      & \multicolumn{1}{c|}{$acc_{{clean}}$} & \multicolumn{1}{c|}{$acc_{{cert}}$} & $acc_{{cert}^2}$ \\ \hline
VOT       & \multicolumn{1}{c|}{71.4}      & \multicolumn{1}{c|}{50.9}         & 30.5            \\ \hline
PC & \multicolumn{1}{c|}{{81.6}}      & \multicolumn{1}{c|}{{63.5}}         & $\oplus$             \\ \hline
MRCert       & \multicolumn{1}{c|}{{81.3}}      & \multicolumn{1}{c|}{57.8}         & {38.7}            \\ \hline\hline
Patch Size   & \multicolumn{3}{c|}{32 pixels}                                                            \\ \hline
Metrics      & \multicolumn{1}{c|}{$acc_{{clean}}$} & \multicolumn{1}{c|}{$acc_{{cert}}$} & $acc_{{cert}^2}$ \\ \hline
VOT       & \multicolumn{1}{c|}{71.4}      & \multicolumn{1}{c|}{41.2}         & 12.1            \\ \hline
PC & \multicolumn{1}{c|}{{81.4}}      & \multicolumn{1}{c|}{{53.4}}         & $\oplus$             \\ \hline
MRCert       & \multicolumn{1}{c|}{{81.3}}      & \multicolumn{1}{c|}{39.5}         & {15.6}            \\ \hline
\multicolumn{4}{l}{Note: $\oplus$=theoretically infeasible}
\end{tabular}
\end{table}

\begin{table*}[]
\caption{Results of certified recovery defenders on
ImageNette and CIFAR10 with patch size 16 and 32 pixels (in \%).}
\captionsetup{font={small}}\label{tab:benign_imagenette_and_cifar}
\resizebox{\linewidth}{!}{
\begin{tabular}{|c||cccccc||cccccc|}
\hline
Dataset      & \multicolumn{6}{c||}{ImageNette}                                                                                                                                                                  & \multicolumn{6}{c|}{CIFAR10}                                                                                                                                                                     \\ \hline
Patch Size   & \multicolumn{3}{c|}{16 pixels}                                                                                 & \multicolumn{3}{c||}{32 pixels}                                                            & \multicolumn{3}{c|}{16 pixels}                                                                                 & \multicolumn{3}{c|}{32 pixels}                                                            \\ \hline
Metrics      & \multicolumn{1}{c|}{$acc_{{clean}}$} & \multicolumn{1}{c|}{$acc_{{cert}}$} & \multicolumn{1}{c|}{$acc_{{cert}^2}$} & \multicolumn{1}{c|}{$acc_{{clean}}$} & \multicolumn{1}{c|}{$acc_{{cert}}$} & $acc_{{cert}^2}$ & \multicolumn{1}{c|}{$acc_{{clean}}$} & \multicolumn{1}{c|}{$acc_{{cert}}$} & \multicolumn{1}{c|}{$acc_{{cert}^2}$} & \multicolumn{1}{c|}{$acc_{{clean}}$} & \multicolumn{1}{c|}{$acc_{{cert}}$} & $acc_{{cert}^2}$ \\ \hline
VOT       & \multicolumn{1}{c|}{96.2}      & \multicolumn{1}{c|}{91.3}         & \multicolumn{1}{c|}{80.5}            & \multicolumn{1}{c|}{96.2}      & \multicolumn{1}{c|}{87.6}         & 57.3            & \multicolumn{1}{c|}{89.3}      & \multicolumn{1}{c|}{76.9}         & \multicolumn{1}{c|}{56.9}            & \multicolumn{1}{c|}{89.3}      & \multicolumn{1}{c|}{68.6}         & {28.6}            \\ \hline
PC & \multicolumn{1}{c|}{{99.5}}      & \multicolumn{1}{c|}{{97.2}}         & \multicolumn{1}{c|}{$\oplus$}             & \multicolumn{1}{c|}{{99.4}}      & \multicolumn{1}{c|}{{93.7}}         & $\oplus$             & \multicolumn{1}{c|}{{98.5}}      & \multicolumn{1}{c|}{{90.0}}         & \multicolumn{1}{c|}{$\oplus$}             & \multicolumn{1}{c|}{{98.2}}      & \multicolumn{1}{c|}{{80.9}}         & $\oplus$             \\ \hline
MRCert       & \multicolumn{1}{c|}{{99.3}}      & \multicolumn{1}{c|}{95.7}         & \multicolumn{1}{c|}{{87.4}}            & \multicolumn{1}{c|}{{98.9}}      & \multicolumn{1}{c|}{89.6}         & {66.7}            & \multicolumn{1}{c|}{{98.1}}      & \multicolumn{1}{c|}{86.2}         & \multicolumn{1}{c|}{{63.3}}            & \multicolumn{1}{c|}{{97.6}}      & \multicolumn{1}{c|}{67.4}         & {28.2}            \\ \hline
\multicolumn{4}{l}{Note: $\oplus$=theoretically infeasible}
\end{tabular}
}
\end{table*}

\subsection{Experimental Results and Data Analysis}\label{sec:Experimental_Results}

\subsubsection{RQ1: How well does MRCert empirically provide a robustness certificate at post-deployment time?}

Table~\ref{tab:attack} summarizes the results of the masking-based defenders (MRCert and PC) on $\mathbb{D}_{pat}$.
For fairness, we only compare PC and MRCert against this empirical attack.

The left of Table~\ref{tab:attack} shows the results when the patch size is 16 pixels.
Maintaining high accuracy across a mixture of benign and adversarial patches simulates the complexity of real-world post-deployment, where defenders must simultaneously handle natural perturbation and powerful attacks.
Although MRCert's $acc_{clean}$ is lower but close to PC's, it achieves much higher $acc_{cert}$ at a 16-pixel patch size by 3.45x, 2.23x, and 1.27x on ImageNet, ImageNette, and CIFAR10, respectively, underscoring its robust certification capability for post-deployment scenarios.
Specifically, MRCert can guarantee the certification on adversarially patched samples of round-trip certified samples, and indeed achieves non-trivial adversarial certified accuracy, while PC gets 0.
By type-specific handling of adversarially patched samples and benign (benignly patched) samples, MRCert attains significantly higher certified accuracy while maintaining clean accuracy comparable to PC, indicating predictions of more inputs are guaranteed to be robust and benign.

PC becomes even worse when the patch size increases to 32 pixels: nearly all patched samples are adversarial.
Consequently, its $acc_{cert}$ vanishes across all datasets, except for a negligible fraction of non-adversarial samples that remain certified on CIFAR-10, hitting the wall in $acc_{cert}$.
Still, PC gets 0 in $acc^{adv}_{cert}$.
Although the clean accuracy of MRCert suffers from a decrease for all datasets,
it still maintains a high $acc_{cert}$ for patch size of 32 pixels,
where $acc_{cert}$ is 56.1\% (=$\nicefrac{24.8}{44.2}$), 85.5\% (=$\nicefrac{80.2}{93.8}$), and 60.5\% (=$\nicefrac{46.8}{77.4}$) of that for 16 pixels, respectively for ImageNet, ImageNette, and CIFAR10, which indicates that the type-specific design in MRCert still works very well against such a tough attack.
While PC exhibits higher $acc_{clean}$, contingent on prior knowledge of the patch count for recovery (one patch), it fails to provide any formal certified guarantees for these inputs.

\subsubsection{RQ2: How well does MRCert provide a theoretical bound for its performance at post-deployment time?}

Following standard practices in prior works, we compare the theoretical lower bounds derived from benign samples to assess the defenders' performance on adversarially patched samples.
$acc_{cert}$ and $acc_{{cert}^2}$
in this RQ characterize the model's ability to issue robustness certificates for benign and adversarial inputs at post-deployment time, respectively.

Tables \ref{tab:benign_imagenet} and \ref{tab:benign_imagenette_and_cifar} summarize the results on benign samples in the three benchmarks.
Note that PC is incapable of providing $acc_{cert^2}$ (its original evaluation does not measure $acc_{cert^2}$, and its fundamental infeasibility is formalized in $\S \ref{sec:moti}$).

\emph{The findings on ImageNet:}
As shown in Table~\ref{tab:benign_imagenet}, MRCert outperforms VOT in $acc_{cert^2}$ by 8.2\% and 3.5\% for the two patch sizes, highlighting its superior capability in certifying adversarial samples.
In $acc_{cert}$, MRCert outperforms VOT by 6.9\% at a 16-pixel patch size, yet underperforms it by 1.7\% at 32 pixels.
Crucially, VOT suffers from a roughly 10\% decrease in clean accuracy compared to MRCert due to its smoothing mechanism without distinction. This degradation not only compromises prediction quality but also enables a more insidious threat: attackers can craft adversarial samples that yield incorrect predictions yet are erroneously granted a robustness certificate since the benign prediction is incorrect.
While PC exhibits competitive $acc_{clean}$ and even superior $acc_{cert}$ relative to MRCert, its theoretical limits prevent certification of adversarial samples in terms of $acc_{cert^2}$.

\emph{The findings on ImageNette and CIFAR10:}
Similar trends are observed across ImageNette and CIFAR-10.
On ImageNette, MRCert consistently outperforms VOT across all metrics ($acc_{clean}$, $acc_{cert}$, and $acc_{cert^2}$), with average improvements of 2.9\%, 3.2\%, and 8.2\%, respectively. This indicates that while VOT's smoothing mechanism has a relatively smaller impact on $acc_{clean}$ for ImageNette, it still significantly degrades its $acc_{cert}$ and $acc_{cert^2}$.
On CIFAR-10, MRCert achieves performance comparable to PC in $acc_{clean}$ and $acc_{cert}$ while providing the additional $acc_{cert^2}$ guarantee.
Notably, with a 16-pixel patch, MRCert surpasses VOT in all metrics. Even as the patch size increases to 32 pixels, MRCert maintains competitive $acc_{cert}$ and $acc_{cert^2}$ (within 1.2\% and 0.4\% of VOT, respectively) while sustaining a significantly higher $acc_{clean}$ (8.3\% over VOT).

\subsubsection{Summary and discussion}
Across all datasets and patch sizes in the experiment, MRCert achieves a superior balance relative to existing baselines.
MRCert is competitive with PC (which requires prior knowledge of patch counts) in $acc_{clean}$ and additionally provides a non-trivial round-trip certification ($acc_{cert^2}$) to guarantee a benign prediction with a certificate at post-deployment time.
Compared to VOT, MRCert consistently yields higher $acc_{clean}$ and competitive or superior certified accuracies ($acc_{cert}$ and $acc_{cert^2}$).
Critically, the empirical results of MRCert for adversarially patched samples consistently exceed bounds from benign samples (e.g., $acc_{cert}$ in RQ1 to $acc_{cert}^2$ in RQ2), validating the reliability of our theory.
Despite the seemingly low round-trip and adversarial certified accuracies demonstrated by MRCert in absolute measure, our approach establishes a novel paradigm for masking-based recovery, laying the foundation for future work.

We only adopt the original finetuning method from \cite{xiang2024patchcure}, where a specific finetuning method like \cite{saha2023adversarial} for \cite{xiang2022patchcleanser} may improve its performance.
Although MRCert generated more mutants than PC, the certification overhead of MRCert to certify the entire ImageNet validation set is 34 hours on a single 3090 GPU, which is still on par with many peers (41 hours in \cite{zhou2024crosscert}, 38-40 hours in \cite{xiang2022patchcleanser}).
The condition for round-trip certification against one patch in MRCert is equivalent to that against two patches in PC.
Our implementation is built on PC's implementation, which is tightly coupled with its customized ViT and BagNet, and the original PC+BagNet fails to certify samples against two patches of any size.
Both VOT and PC fail to certify any $224 \times 224$ samples against two patches of $48 \times 48$ pixels or larger.
We leave the acceleration of MRCert's implementation, its generalization to broader base models, and its certification against more and larger adversarial patches for future
work.

\section{Related Work}
Recently, adversarial patch attacks have gained increasing attention \cite{wei2024physical} as a significant threat to real-world systems such as commercial autonomous vehicles \cite{wang2025revisiting}, as they model physical-world adversarial objects like a sticker \cite{wei2022adversarial} or a tag \cite{liu2020bias}.
Experimental results \cite{chiang2020certified,li2023sok} reveal the susceptibility of empirical defenses \cite{hayes2018visible,naseer2019local}.
Consequently, certified defenses have emerged to offer deterministic provable guarantees.
Certified recovery ensures the recovery of the benign label of adversarially patched samples.
To overcome the scalability bottlenecks of IBP \cite{chiang2020certified}, extensive research has pivoted toward smoothing \cite{xiang2021patchguard,salman2022certified,li2022vip,zhou2023majority}, pioneered by Derandomized Smoothing \cite{levine2020randomized}, which demonstrates its advancement against those randomized \cite{cohen19,levine2020robustness}.
\cite{yatsura2023certified} adapts Derandomized Smoothing to semantic segmentation and proposes post-deployment certification.
\cite{zhou2024crosscert} cross-checks different recovery defenders for detection.
Masking-based certified detection \cite{li2022vip,patchcensor,zhou2026hicert} has emerged, leveraging masking to guarantee the detection of adversarially patched samples, which cannot provide any guarantee at post-deployment time \cite{yatsura2023certified}.
\cite{xiang2022patchcleanser} further detects their detection outcomes, extending the paradigm to the certified recovery of the benign prediction while preventing the erosion caused by the smoothing process.
\cite{saha2023revisiting} strengthens its training, and
\cite{xiang2024patchcure} improve its efficiency.
However, they all fail to provide any certificate for adversarially patched samples.
Our work represents the first attempt to provide a robustness certificate for both benign and adversarially patched samples in line with this SOTA.

\section{Conclusion}
This work proposes MRCert, the \emph{first} masking-based defender to provide certified recovery at post-deployment time for adversarially patched samples, while simultaneously maintaining high robustness and clean accuracy validated through theorems and experiments.
Our replication package, including the supplementary document (appendix), can be found in \url{https://anonymous.4open.science/r/mrcert/}.

\clearpage
\bibliographystyle{named}
\bibliography{references_short}

\appendix
\clearpage
\thispagestyle{empty}
\twocolumn[{\centering \Large Supplementary Document \\ for \\\textbf{MRCert: Towards Post-deployment Patch Robustness Certification for Adversarially Patched Samples via Type-specific Masking}\par\vspace{1em}}]
\section{Formal Proof}

\begin{Property*}[Infeasibility of Re-certifying Adversarially Patched Samples for Existing Masking-based Recovery Works]
$
\forall x' \in \mathbb{A}_{\mathbb{P}}(x), c_r(x) \land f(x') \neq f(x) \implies \neg c_r(x').
    $
\end{Property*}
\begin{proof}
We know in \cite{xiang2022patchcleanser}, $c_r(\hat{x}):=[\forall {\textsc{m}_1},{\textsc{m}_2}\in\mathbb{M}_\mathbb{P}, f(\hat{x}\odot{\textsc{m}_1}\odot{\textsc{m}_2})= f(\hat{x})]$.
Then if $c_r(x)=\textit{True}$, we know
$
[\forall {\textsc{m}_1},{\textsc{m}_2}\in\mathbb{M}_\mathbb{P}, f({x}\odot{\textsc{m}_1}\odot{\textsc{m}_2})= f({x})]
$.
We also know $\forall \textsc{p}\in\mathbb{P}, \exists \textsc{m}\in\mathbb{M}_\mathbb{P},
\textsc{m}\odot\textsc{p}=\textsc{p}
$ by definition, which makes
$\forall x'\in\mathbb{A}_\mathbb{P}(x), x'\odot\textsc{m}=x\odot\textsc{m}$ (see Fig.~\ref{fig:mask-overview_concept}).
Therefore,
$\exists x'\in\mathbb{A}_\mathbb{P}(x), f(x'\odot\textsc{m})=f(x\odot\textsc{m})=f(x)$.
We have $f(x')\neq f(x)$ from the antecedent; therefore,
$\exists x'\in\mathbb{A}_\mathbb{P}(x), f(x'\odot\textsc{m})\neq f(x')$.
Since
$c_r({x}')$ requires $[\forall {\textsc{m}_1},{\textsc{m}_2}\in\mathbb{M}_\mathbb{P}, f({x}'\odot{\textsc{m}_1}\odot{\textsc{m}_2})= f({x}')]$,
$c_r(x')=\textit{False}$ holds.
\end{proof}

\begin{lem*}[Certification of benign samples]
\label{lem:recovery_once_attack}

Given
an arbitrary sample $\hat{x}$,
if $[\forall {\textsc{m}_1},{\textsc{m}_2},{\textsc{m}_3}\in\mathbb{M}_\mathbb{P}, f(\hat{x}\odot{\textsc{m}_1}\odot{\textsc{m}_2}\odot{\textsc{m}_3})= f(\hat{x})]$ holds,
$\forall \hat{x}'\in\mathbb{A}_\mathbb{P}(\hat{x}),
g(\hat{x}')=g(\hat{x})
$.
\end{lem*}

\begin{proof}
We first analyze $g(\hat{x})$.
Since
$[\forall {\textsc{m}_1},{\textsc{m}_2},{\textsc{m}_3}\in\mathbb{M}_\mathbb{P}, f(\hat{x}\odot{\textsc{m}_1}\odot{\textsc{m}_2}\odot{\textsc{m}_3})= f(\hat{x})]
\implies
[\forall {\textsc{m}_1},{\textsc{m}_2}\in\mathbb{M}_\mathbb{P}, f(\hat{x}\odot{\textsc{m}_1}\odot{\textsc{m}_2})= f(\hat{x})]
$,
we know its prediction label
$g(\hat{x})=f(\hat{x})$ output in Case \textcircled{1}.
We then analyze $g(\hat{x}')$.
For $\forall \hat{x}'\in\{\hat{x}' \mid \hat{x}'=(\textsc{J}-\textsc{p}){\odot} \hat{x}+\textsc{p}{\odot} \hat{x}' \land \textsc{p}\in \mathbb{P} \}$ (i.e., $\forall \hat{x}'\in\mathbb{A}_\mathbb{P}(\hat{x})$),
w.l.o.g.,
we let ${\textsc{m}_1}\odot\textsc{p}=\textsc{O}$.
Therefore, we get $\hat{x}'\odot{\textsc{m}_1}=((\textsc{J}-\textsc{p}){\odot} \hat{x}+\textsc{p}{\odot} \hat{x}')\odot{\textsc{m}_1}=\hat{x}\odot{\textsc{m}_1}$.
Then, we get
$[\exists {\textsc{m}_1}\in\mathbb{M}_\mathbb{P}, \forall {\textsc{m}_2}, {\textsc{m}_3}\in\mathbb{M}_\mathbb{P}, f(\hat{x}'\odot{\textsc{m}_1}\odot{\textsc{m}_2}\odot{\textsc{m}_3})= f(\hat{x})]$ (see Fig.~\ref{fig:mask-overview_concept}).
Note that the special cases ${\textsc{m}_2}={\textsc{m}_3}$, ${\textsc{m}_1}={\textsc{m}_2}={\textsc{m}_3}$ are included.
Case 1: Suppose the returned label of $\hat{x}'$ output in Case \textcircled{1} as $f(\hat{x}')$ (i.e., $\forall {\textsc{m}_1}, {\textsc{m}_2}\in\mathbb{M}_\mathbb{P}, f(\hat{x}'{\odot} {\textsc{m}_1} {\odot} {\textsc{m}_2})= f(\hat{x}')$),
then we know $g(\hat{x}')=f(\hat{x}')=f(\hat{x})=g(\hat{x})$.
Case 2: Otherwise, its returned label should be output in Case \textcircled{2} as
$f(\hat{x}'\odot{\textsc{m}_1})$, since $[ \exists{\textsc{m}_1}\in\mathbb{M}_\mathbb{P},\forall {\textsc{m}_2}, {\textsc{m}_3}\in\mathbb{M}_\mathbb{P}, f(\hat{x}'\odot{\textsc{m}_1}\odot{\textsc{m}_2}\odot{\textsc{m}_3})=f(\hat{x}'{\odot} \textsc{m}_1)=f(\hat{x})]$, and further $g(\hat{x}')=f(\hat{x}'\odot{\textsc{m}_1})=f(\hat{x})=g(\hat{x})$.
\end{proof}

\begin{lem*}[Certification of adversarially patched samples]
Given an arbitrary sample $\hat{x}$,
if $[\exists {\textsc{m}_1}\in\mathbb{M}_\mathbb{P},\forall$ ${\textsc{m}_2}$,${\textsc{m}_3}$,${\textsc{m}_4}\in\mathbb{M}_\mathbb{P}, f(\hat{x}\odot\textsc{m}_1\odot\textsc{m}_2\odot\textsc{m}_3\odot\textsc{m}_4)= f(\hat{x}\odot{\textsc{m}_1})]$ holds,
$\forall \hat{x}'\in\mathbb{A}_\mathbb{P}(\hat{x}),
g(\hat{x}')=g(\hat{x})
$.
\end{lem*}

\begin{proof}
We first analyze $g(\hat{x})$.
Case 1: If $\hat{x}$ meet the condition of Case \textcircled{1},
we can further get
$\forall {\textsc{m}_1}, {\textsc{m}_2}, {\textsc{m}_3}, {\textsc{m}_4}\in\mathbb{M}_\mathbb{P}, f(\hat{x}\odot\textsc{m}_1\odot\textsc{m}_2\odot{\textsc{m}_3}\odot{\textsc{m}_4})= f(\hat{x})$.
Note that the special case $\textsc{m}_1=\textsc{m}_2=\textsc{m}_3=\textsc{m}_4$ is included,
which means we can get $g(\hat{x})=f(\hat{x}{\odot} \textsc{m}_1)$.
Case 2: Otherwise, by the given condition $[\exists {\textsc{m}_1}\in\mathbb{M}_\mathbb{P},\forall$ ${\textsc{m}_2}$,${\textsc{m}_3}$,${\textsc{m}_4}\in\mathbb{M}_\mathbb{P}, f(\hat{x}\odot\textsc{m}_1\odot\textsc{m}_2\odot\textsc{m}_3\odot\textsc{m}_4)= f(\hat{x}\odot{\textsc{m}_1})]\implies
[\exists {\textsc{m}_1}\in\mathbb{M}_\mathbb{P},\forall$ ${\textsc{m}_2}$,${\textsc{m}_3}$$\in\mathbb{M}_\mathbb{P}, f(\hat{x}\odot\textsc{m}_1\odot\textsc{m}_2\odot\textsc{m}_3)= f(\hat{x}\odot{\textsc{m}_1})]
$, $\hat{x}$ meet the condition in Case \textcircled{2} and its prediction label $g(\hat{x})=f(\hat{x}{\odot} \textsc{m}_1)$, same as Case \textcircled{1}.
We then analyze $g(\hat{x}')$.
For $\hat{x}'\in\{\hat{x}' \mid \hat{x}'=(\textsc{J}-\textsc{p}){\odot} \hat{x}+\textsc{p}{\odot} \hat{x}' \land \textsc{p}\in \mathbb{P} \}$ (i.e., $\forall \hat{x}'\in\mathbb{A}_\mathbb{P}(\hat{x})$),
Case 1: Suppose $\textsc{m}_1{\odot} \textsc{p}=\textsc{O}$.
Then we can get
$\hat{x}'\odot{\textsc{m}_1}=((\textsc{J}-\textsc{p}){\odot} \hat{x}+\textsc{p}{\odot} \hat{x}')\odot{\textsc{m}_1}=\hat{x}\odot{\textsc{m}_1}$,
and further get
$[\forall {\textsc{m}_2}, {\textsc{m}_3}, {\textsc{m}_4}\in\mathbb{M}_\mathbb{P}, f(\hat{x}'\odot\textsc{m}_1\odot\textsc{m}_2\odot\textsc{m}_3\odot\textsc{m}_4)= f(\hat{x}\odot{\textsc{m}_1})]$, which is the same condition as that on $\hat{x}$.
Therefore, repeating those analysis for $g(\hat{x})$ above can get
$g(\hat{x})=g(\hat{x}')$.
Case 2: Otherwise, for ${\textsc{m}_2}, {\textsc{m}_3}, {\textsc{m}_4}$,
w.l.o.g., we let ${\textsc{m}_2}\odot\textsc{p}=\textsc{O}$ ($\textsc{m}_1\neq\textsc{m}_2$).
Then similarly, we get
$[\exists {\textsc{m}_1, \textsc{m}_2 (\neq \textsc{m}_1)}\in\mathbb{M}_\mathbb{P}, \forall {\textsc{m}_3}, {\textsc{m}_4}\in\mathbb{M}_\mathbb{P}, f(\hat{x}'\odot{\textsc{m}_1}\odot{\textsc{m}_2}\odot{\textsc{m}_3}\odot{\textsc{m}_4})= f(\hat{x}{\odot} \textsc{m}_1)]$.
Note that the special cases ${\textsc{m}_1}={\textsc{m}_3}, {\textsc{m}_2}={\textsc{m}_4}$,${\textsc{m}_3}={\textsc{m}_4}$  are included.
Case 2.1: Suppose the prediction label $g(\hat{x}')$ output in Case \textcircled{1}.
Then, since
$[\exists {\textsc{m}_1, \textsc{m}_2 (\neq \textsc{m}_1)}\in\mathbb{M}_\mathbb{P}, f(\hat{x}'\odot{\textsc{m}_1}\odot{\textsc{m}_2})= f(\hat{x}{\odot} \textsc{m}_1)]$ (special case $\textsc{m}_1=\textsc{m}_3$, $\textsc{m}_2=\textsc{m}_4$), by the condition of Case \textcircled{1}, we know
$g(\hat{x}')=f(\hat{x}')=f(\hat{x}{\odot} \textsc{m}_1)=g(\hat{x})$.
Case 2.2: Suppose the prediction label $g(\hat{x}')$ output in Case \textcircled{2}.
Then, since
$[\exists {\textsc{m}_1, \textsc{m}_2 (\neq \textsc{m}_1)}\in\mathbb{M}_\mathbb{P}, \forall \textsc{m}_3\in\mathbb{M}_\mathbb{P}, f(\hat{x}'\odot{\textsc{m}_1}\odot{\textsc{m}_2}\odot{\textsc{m}_3})= f(\hat{x}{\odot} \textsc{m}_1)]$ (special case $\textsc{m}_3=\textsc{m}_4$), by the condition of Case \textcircled{2}, we know
$g(\hat{x}')=f(\hat{x}'{\odot} \textsc{m}_1)=f(\hat{x}{\odot} \textsc{m}_1)=g(\hat{x})$.
Case 2.3: Otherwise,
the prediction label of $\hat{x}'$ should output in Case \textcircled{3}
since
$[\exists {\textsc{m}_1, \textsc{m}_2 (\neq \textsc{m}_1)}\in\mathbb{M}_\mathbb{P}, \forall {\textsc{m}_3}, {\textsc{m}_4}\in\mathbb{M}_\mathbb{P}, f(\hat{x}'\odot{\textsc{m}_1}\odot{\textsc{m}_2}\odot{\textsc{m}_3}\odot{\textsc{m}_4})=f(\hat{x}{\odot} \textsc{m}_1)= f(\hat{x}'{\odot} \textsc{m}_1\odot\textsc{m}_2)]$ (special case $\textsc{m}_1=\textsc{m}_3$, $\textsc{m}_2=\textsc{m}_4$),
and further $g(\hat{x}')=f(\hat{x}'{\odot} \textsc{m}_1\odot\textsc{m}_2)=f(\hat{x}{\odot} \textsc{m}_1)=g(\hat{x})$.
\end{proof}

\begin{thm*}[Certification of samples]
Given an arbitrary sample $\hat{x}$,
if $c(\hat{x})=True$ holds,
$\forall \hat{x}'\in\mathbb{A}_\mathbb{P}(\hat{x}),
g(\hat{x}')=g(\hat{x})
$.
\end{thm*}
Simply conjoining the antecedent of Lemma~\ref{lem:recovery_benign} and Lemma~\ref{lem:recovery_one_patched_sample} can prove this theorem.

\begin{thm*}[Round-trip certification of samples]
Given a benign sample ${x}$, if
$[\forall {\textsc{m}_1},{\textsc{m}_2},{\textsc{m}_3},{\textsc{m}_4}\in\mathbb{M}_\mathbb{P}, f({x}\odot$
$\textsc{m}_1$$\odot$$\textsc{m}_2$$\odot\textsc{m}_3$$\odot\textsc{m}_4)=f(x)]$
(i.e., $c_r^2(x)=\textit{True}$),
$[\forall {x}'\in\mathbb{A}_\mathbb{P}({x}),
g({x}')=g({x})
\land
c_r({x}')=\textit{True}
\land
c_r(x)=\textit{True}]
$.
\end{thm*}

\begin{proof}
By the condition $[\forall {\textsc{m}_1},{\textsc{m}_2},{\textsc{m}_3},{\textsc{m}_4}\in\mathbb{M}_\mathbb{P}, f({x}\odot$
$\textsc{m}_1$$\odot$
$\textsc{m}_2$
$\odot\textsc{m}_3$
$\odot\textsc{m}_4)=f(x)]$, we know the returned label $g({x})=f(x)$ in Case \textcircled{1} and $c_r(x)=\textit{True}$.
Still by this condition, we know
[$\forall x'\in\mathbb{A}_\mathbb{P}(x), \exists \textsc{m}_1\in\mathbb{M}_\mathbb{P},
\forall {\textsc{m}_2},
{\textsc{m}_3},
{\textsc{m}_4}\in\mathbb{M}_\mathbb{P},
f({x}'\odot$
$\textsc{m}_1$$
\odot$$\textsc{m}_2$
$\odot\textsc{m}_3$
$\odot\textsc{m}_4)=f(x)$], since
[$\exists \textsc{m}_1\in\mathbb{M}_\mathbb{P}, $$\forall $
$x'\in\mathbb{A}_\mathbb{P}(x), x'\odot\textsc{m}_1=x{\odot} \textsc{m}_1$] (see Fig.~\ref{fig:mask-overview_concept} for illustration).
Then by Lemma~\ref{lem:recovery_one_patched_sample}, we know [$\forall x'\in\mathbb{A}_\mathbb{P}(x), c_r(x')=\textit{True}$].
By Lemma~\ref{lem:recovery_benign}, we also know [$\forall x'\in\mathbb{A}_\mathbb{P}(x),$
$g(x')=g(x)$].
Finally, we know
$[\forall {x}'\in\mathbb{A}_\mathbb{P}({x}),
g({x}')=g({x})
\land
c_r({x}')=\textit{True}\land
c_r({x})=\textit{True}]
$.
\end{proof}

\section{The Subtlety in the Design of MRCert}
\emph{The ordering of the type-specific cases in $g$.}
The sequence of the cases in $g$ is crucial. Changing the order will either break the certification guarantee for adversarially patched samples or degrade the clean accuracy.
For example, swapping Case \textcircled{2} and Case \textcircled{3} would require the benign counterpart of an adversarially patched sample to exhibit the same prediction labels under four masks (instead of three) to remain certified; otherwise, a patch attacker could exploit a fourth‑order masked mutant that predicts another label to evade the detection of the adversarially patched sample by this case-swapped version of $g$.
In our proposed version of $g$,
Case 1 is intentionally placed first to prevent benign samples from reaching the recovery stages, thereby improving clean accuracy.

\emph{The ordering of mask conditions within a type-specific case in $g$.}
The internal ordering of masking conditions also matters. Consider Case 2 (indirect testing): Suppose that its quantifiers were reordered from
``there exists mask 1 such that for all mask 2, mask3...''
to
``there exist mask 1 and mask 2 such that for all mask 3 ...''.
Then, adversarial samples containing two patches could escape certification, even when their benign counterparts are round-trip certified.
In this reordered structure, the adversarial patched sample could simply choose two masks (mask 1 and mask 2) that expose both patches, i.e., they do not cover either patch. Then, even if mask 3 covers one patch, the other patch remains uncovered, allowing the attacker to arbitrarily alter the prediction and evade detection.
This shows that the original quantifier ordering is essential for preventing such escape cases.

\emph{The number of masks within a type-specific cases in $g$.}
Taking Case 2 as an example, the number of masks is also critical.
Reducing the number of masks in Case 2 would allow adversarial samples with two patches to escape from the detection of indirect testing,
for reasons analogous to the ordering issue discussed above.
Conversely, increasing the number of masks (e.g., by 1) would increase the masking requirements for certification, which makes certifying a benign sample require prediction consistency under four masks, instead of the current three, which directly decreases the total number of samples that can be certified.

\section{Extension of MRCert to Recover and Certify Samples with N-patches}
We can extend the maximum number of patches from 2 to $N$ (called MRCert-N-patch, a variant of MRCert) following the following idea.
First, we apply each set of $N$ masks in the covering mask set $\mathbb{M}$ on the input sample $\hat{x}$ to test whether $\hat{x}$ is harmful.
If it is not harmful, MRCert-N-patch returns the label $f(\hat{x})$ (marked as N-Case \textcircled{1}).
If $\hat{x}$ is detected as harmful, we then test whether all its first-order mutants are harmful by applying each possible subset with $N$ masks selected with replacement from $\mathbb{M}$ on each first-order mutant of $\hat{x}$.
If there exists a first-order mutant, whose all $(N+1)$th-order mutants generated from the first-order mutant of $\hat{x}$ are predicted with the same label as this first-order sample, then  $\hat{x}$  is deemed as a one-patch harmful sample, this first-order mutant is ``clean'' and the prediction label of this first-order mutant is returned (marked as N-Case \textcircled{2}).
If that is not the case, we then test whether all second-order mutants of  $\hat{x}$ are harmful in the same manner, and repeat until the $N$th-order mutants of  $\hat{x}$ are tested.
For the certification function $c_r$ with the input sample $\hat{x}$, it should be extended to the condition that all $(N+1)$th-order mutants of $\hat{x}$ are predicted with the same label as $\hat{x}$ (for those input samples whose label returned in Case \textcircled{1}), and the condition that there exists a first-order mutant of $\hat{x}$, whose all $(N+2)$th-order mutants are predicted with the same label as this first-order mutant of $\hat{x}$ (for those input samples whose label returned in Case \textcircled{2}), and certifying the input samples output in other cases by the condition in the same manner.
For the round-trip certification function, it should be the condition that all $2N$th-order mutants of a benign sample $x$ are predicted with the same label as $x$.
We leave the formal proof and implementation as future work.

\section{Experimental Setup}
\subsection{Environment and Dataset}
The evaluation is conducted on an Ubuntu 20.04 machine equipped with four Nvidia 3090 GPUs.
Following \cite{xiang2022patchcleanser}, we adopt 1000-class ImageNet \cite{deng2009imagenet}, 10-class CIFAR10 \cite{krizhevsky2009learning}, and 10-class ImageNette as our datasets, encompassing both large-scale, diverse datasets and efficient benchmarks widely used for rapid experimentation.

\subsection{Models and Baselines}
We adopt the Vision Transformer (ViT) as our base model, which achieves the state-of-the-art in many patch robustness certification defenders \cite{xiang2024patchcure,li2022vip,salman2022certified,xiang2022patchcleanser,patchcensor}.

We adopt the state-of-the-art masking-based defender
PatchCURE (\textbf{PC}) \cite{xiang2024patchcure} and state-of-the-art smoothing-based defender
VOT-CrossCert (\textbf{VOT}) from CrossCert \cite{zhou2024crosscert}, which shares the common methodology for smoothing-based recovery with ViP \cite{li2022vip} and S-ViT \cite{salman2022certified}.
We arm VOT with the methodology for runtime verification proposed by \cite{yatsura2023certified}
as our baseline to compare with \textbf{MRCert} in our infrastructure.

Specifically, MRCert adopts the end-to-end ViT-SRF model (setting 14x1-k6\footnote{Other variants ViT-SRF14x2, ViT-SRF2x2, and BagNet33 cannot handle the situation against two patches, since all their receptive fields would be inherently masked by VIT-SRF \cite{xiang2024patchcure} when generating corresponding mutants. $k$ is the parameter to control the position of splitting SRF and LRF in VIT-SRF in the range [0, 12], which can tune the trade-off between computational efficiency and robustness.
We adopt the middle one since it is not the focus of this paper.}) proposed and used by PatchCURE \cite{xiang2024patchcure} from \url{https://github.com/inspire-group/PatchCURE} with the pre-train weights and setting from MAE \cite{he2022masked}, so that
MRCert can use the same strategy in the experiment as PatchCURE to calculate the covering mask set, generate mutants with their predictions inside the ViT-SRF.
We use the training scripts with the settings and hyperparameters from the official repository of PatchCURE for fine-tuning ViT-SRF, which uses the training samples with their first and second-order mutants only.
We then apply the same finetuned ViT-SRF to both MRCert and PatchCURE to ensure fairness.
We also adopt this official repository \url{https://github.com/inspire-group/PatchCURE} to implement PatchCURE in our infrastructure.
We adopt the same pre-trained weights and settings from MAE \cite{he2022masked} for VOT.
We follow the fine-tuning settings and hyperparameters from the state-of-the-art S-ViT \cite{salman2022certified} for VOT (note that smoothing-based recovery adopts a different notion of mutants, where a mutant is a slice of a sample; therefore, a different fine-tuning on the base model is needed \cite{li2022vip}).
We also adopt the column ablation with the ablation size of 19 pixels from \cite{salman2022certified}.
We have extracted the corresponding results from the original papers of ViP \cite{li2022vip} and S-ViT \cite{salman2022certified} to compare with VOT in Tab.~\ref{tab:smooth_results}
(their results on round-trip certification are not reported in their papers).

\subsection{Metrics}
\label{sec:metrics}
Let ${x}$ be a benign sample with the true label $y_0$ in a clean test dataset $\mathbb{D}$, and $R = \langle g(x), c(x),  c^2_r(x)\rangle$ be a certified recovery defender.
 \textbf{Clean accuracy} is the fraction of $\mathbb{D}$ that are correctly predicted,
defined as
$acc_{clean}=\frac{\mid\{{x}\in\mathbb{D}\mid g({x})=y_0\}\mid}{\mid\mathbb{D}\mid}$, which evaluates the standard performance of certified recovery as a classifier.
\textbf{Certified accuracy} is the fraction of $\mathbb{D}$ that are correctly predicted and certified robust, whose adversarially patched samples should be predicted with the same label as the benign sample,
defined as
$acc_{cert}=\frac{\mid\{{x}\in\mathbb{D}\mid g({x})=y_0\land c_r(x)=\textit{True}\}\mid}{\mid\mathbb{D}\mid}$.
\textbf{Round-trip certified accuracy} is the fraction of $\mathbb{D}$ that are correctly predicted and round-trip certified, whose adversarially patched samples should be predicted with the same label as the benign sample and gain a provable robust verdict,
defined as
$acc_{cert^2}=\frac{\mid\{{x}\in\mathbb{D}\mid g({x})=y_0\land c_r^2(x)=\textit{True}\}\mid}{\mid\mathbb{D}\mid}$.
The clean accuracy of
the base classification model $f$ is
$\frac{\mid\{{x}\in\mathbb{D}\mid f({x})=y_0\}\mid}{\mid\mathbb{D}\mid}$.
Clean accuracy and certified accuracy are commonly adopted by peers \cite{levine2020randomized,salman2022certified,li2022vip,xiang2024patchcure,xiang2022patchcleanser,xiang2021patchguard}, and round-trip certified accuracy is proposed by \cite{yatsura2023certified}.

In RQ1, the test dataset $\mathbb{D}$ is a set of patched samples from attacking a subset of the test dataset (validation dataset for ImageNet), denoted as $\mathbb{D}_{pat}$.
Clean accuracy and certified accuracy are measured in this set of patched samples to evaluate certified recovery defenders on their ability to recover predictions and recover predictions with a provably robust verdict, respectively.
We  measure each defender's ability to certify adversarially patched samples by  \textbf{adversarial certified accuracy}
$acc_{cert}^{adv}=\nicefrac{\mid\{\hat{x}\in\mathbb{D}_{pat}\mid g(\hat{x})= y_0\land c_r(\hat{x})=\textit{True} \land f(\hat{x})\neq y_0\}\mid}{\mid\{\hat{x}\in\mathbb{D}_{pat}\mid f(\hat{x})\neq y_0\}\mid}$.
In RQ2, we use the original test dataset (validation dataset for ImageNet) as $\mathbb{D}$.

\subsection{Experimental Procedure}

We adopt the patch sizes of 16 pixels (measured as a square patch region with a side length of 16 pixels) and 32 pixels for all three datasets. All samples are rescaled to 224x224 \cite{zhou2024crosscert,li2022vip,xiang2022patchcleanser,xiang2024patchcure}.

In RQ1, we perform an actual adversarial patch attack IFGSM adopted from \cite{levine2020randomized} on PC and MRCert, which uses their shared base model as the attack model for gradient-based attacks without the knowledge of the non-differentiable label recovery function for fairness.
Note that the recent proposed attacks are more toward practical, such as limiting the access times or can only get the returned prediction label, while our IFGSM attack has direct and full access to the base model, which is a powerful exam against attackers.
We set 80 random starts, 150 iterations per random start, and a step size of 0.05 following \cite{levine2020randomized}.
We randomly select 500 test samples from each test dataset for attacks.
We follow the practice in \cite{levine2020randomized} to return the worst patched sample
for each benign sample and place it into $\mathbb{D}_{pat}$.
In RQ2, we directly adopt the entire test dataset (validation dataset for ImageNet) as $\mathbb{D}$.

In both RQ1 (as a follow-up to the attack) and RQ2, for each defender on each sample in each test dataset,
we first generate and evaluate corresponding mutants using the corresponding base model (PC and MRCert share the same ones), then apply each defender to the prediction results collected for these mutants, and measure the metric values.

\end{document}